\documentclass[11pt]{article}
\usepackage{epsfig}
\usepackage{latexsym}
\usepackage{amsmath}
\usepackage{verbatim}
\usepackage{enumerate}
\usepackage{setspace}
\usepackage{comment}
\usepackage{anysize}
\usepackage{epic}
\usepackage{eepic}
\usepackage{xcolor}
\usepackage{booktabs}

\usepackage[all,color,frame,import]{xy}
\newcommand{\tab}{\hspace{.25in}}

\newtheorem{theorem}{Theorem}[section]
\newtheorem{lemma}{Lemma}[section]
\newtheorem{proposition}{Proposition}[section]

\newtheorem{observation}{Observation}[section]

\newcommand\QED{\ifhmode\allowbreak\else\nobreak\fi
\quad\nobreak$\Box$\medbreak}

\newcommand{\proofstart}{\par\noindent \emph{Proof:} }
\newcommand{\proofend}{\QED\par}
\newenvironment{proof}{\proofstart}{\proofend}

\usepackage[letterpaper,hmargin=1in,vmargin=1in]{geometry}

\newcommand{\peqnp}{\mbox{{P}~$=$~{NP}}}

\long\gdef\boxit#1{\vspace{5mm}\begingroup\vbox{\hrule\hbox{\vrule\kern5pt
\vbox{\kern5pt#1\kern5pt}\kern0pt\vrule}\hrule}\endgroup}

\begin{document}

\title{Speedup in the Traveling Repairman Problem with Unit Time Windows}
\author{Greg N. Frederickson \and Barry Wittman}

\maketitle

\begin{abstract}
The input to the unrooted traveling repairman problem is an undirected metric graph and a subset of nodes, each of which has a time window of unit length.  Given that a repairman can start at any location, the goal is to plan a route that visits as many nodes as possible during their respective time windows.  A polynomial-time bicriteria approximation algorithm is presented for this problem, gaining an increased fraction of repairman visits for increased speedup of repairman motion.
For speedup $s$, we find a $6\gamma/(s + 1)$-approximation for $s$ in the range $1 \leq s \leq 2$ and a $4\gamma/s$-approximation for $s$ in the range $2 \leq s \leq 4$, where $\gamma = 1$ on tree-shaped networks and $\gamma = 2 + \epsilon$ on general metric graphs.
\end{abstract}

\section{Introduction}

When planning the route for an agent who needs to drop off supplies, make repairs, deliver important information, or other similar tasks, distance traveled may not be the only important metric.  The usefulness of visiting a particular location may be heavily dependent on time.  Consider the example of a client who schedules a window of time during which he or she will be home.  An appliance repairman will be able to perform repairs only during that time window.

In the last decade, the algorithms community has achieved much progress with time-sensitive routing problems of this kind.  These problems typically identify the locations to be visited and the cost of traveling between them as the nodes and edges, respectively, of a weighted, undirected graph.  For example, the orienteering problem \cite{Arkin, Bansal, Blum3, Chekuri2, Chen} seeks to find a path that visits as many nodes as possible before a global time deadline.  The deadline traveling salesman problem \cite{Bansal} generalizes this problem further by allowing each location to have its own deadline.  This problem can be generalized even further to the traveling repairman problem \cite{Bansal,Bar-Yehuda,Chekuri,Frederickson3,Karuno3,Tsitsiklis} by allowing each location to have a time window during which it must be visited to receive credit.

In \cite{Frederickson3}, we introduced the first polynomial-time algorithms that give constant approximations to the traveling repairman problem with unit-time windows whenever the underlying graph is a tree or a metric graph.  In this version of the problem, a time window is identified with each {\em service request} whose location is given by the node with which it is associated.  Let a repairman earn a specified \emph{profit} for visiting a service request during its time window.  A planned sequence of visits made by traveling at a given speed is called a \emph{service run}.  The goal of the traveling repairman problem is to plan a service run that maximizes profit.  Unlike much of the preceding literature \cite{Bansal, Bar-Yehuda}, we considered the unrooted version of the problem, in which the repairman may start at any time from any location and stop similarly.  Although the unrooted problem is no harder than the rooted problem, which specifies a starting point, it is a fascinating and difficult problem in its own right.  Both problems are NP-hard when the graph is a tree \cite{Frederickson3} and APX-hard for a general metric graph \cite{Frederickson6}.

As a counterpoint to the repairman problem, we also introduced the speeding deliveryman problem in \cite{Frederickson3, Frederickson6}, with an alternative optimization paradigm, namely speedup.  The input to the speeding deliveryman problem is the same as the input to the traveling repairman problem, but the goal is to find the minimum speed necessary to visit \emph{all} service requests during their time windows and thus collect all profit.
The deliveryman problem is similarly hard,
and in \cite{Frederickson3, Frederickson6} we introduced  polynomial-time approximation algorithms for it.

In both the repairman and deliveryman problems, our algorithms from \cite{Frederickson3, Frederickson6} 
rely on \emph{trimming} windows so that the resulting time windows are pairwise either identical or non-overlapping.  To trim time windows, we first make divisions in time every $.5$ time units, starting at time 0.  We call the time interval that starts at a particular division and continues up to the next division a \emph{period}.  Because we allow no window to start on a period boundary, each time window will completely overlap exactly one period and partially overlap the two neighboring periods.  Trimming then removes those parts of each window that fall outside of the completely overlapped period.

This process of trimming incurs a penalty in our approximations that is a reduction by a factor of $1/3$ in the number of requests serviced by the repairman and an increase by a factor of 4 in the speedup needed by the deliveryman.
Yet one might expect that a spectrum of performance is possible between these two extremes.
If we have any particular speedup $s$ greater than 1 but less than 4,
we expect an increase in the number of serviced requests,
proportionate in some sense to $s$.
Indeed, this is so, as we show in this paper.

The concept of speedup has been used in the scheduling community since its introduction in \cite{Kalyanasundaram}.  Subsequent work in \cite{Phillips} refined the technique and gave it the name \emph{resource augmentation} as well as notation suitable for scheduling.  Results in \cite{Bansal2} have made notable strides by producing $O(1)$-speed $O(1)$-approximation algorithms for some kinds of scheduling problems.  In other words, by increasing the speed of the machines on which the jobs are run by a constant factor, the resulting schedule can achieve a measure of performance within a constant factor of the optimal schedule run on machines at a baseline speed.  Although we do not use the same techniques or notation given in these papers, we note the connection between scheduling and time-constrained routing and the value of resource augmentation when finding approximation algorithms for each class of problem.

For unit-time windows
we present an algorithm that finds approximations parameterized by speedup $s$ and a factor $\gamma$, where $\gamma = 1$ for a tree and $\gamma$ is no more than $2 + \epsilon$ for a metric graph.  These values for $\gamma$ are explained in Sect.~\ref{section:algorithms for repairman}.
For $s$ in the range $1 \leq s \leq 2$, our algorithms find a $6\gamma/(s + 1)$-approximation.
For $s$ in the range $2 \leq s \leq 4$, our algorithms find a $4\gamma/s$-approximation.
As long as unit speed is no greater than the speed needed for an optimal service run to collect all possible profit, our results will hold.  Were the unit speed considerably faster than is needed to service all requests, the idea of offsetting the sub-optimality of an approximation by increasing speed beyond that of an optimal run becomes meaningless.

Our analysis depends upon the interactions between speedup and trimming, but setting the starting point for partitioning periods at time 0 as we did in \cite{Frederickson3} may not result in the most profitable run over trimmed windows.  Especially in the case of non-integer speedups, careful study suggests that some particular starting point for periods may be better than others.  If we repeatedly choose different starting points, run our repairman algorithm on trimmed windows, and always keep the best answer found so far, we can establish better bounds on the final approximation.  Although there are an infinite number of points to start partitioning periods from, we can find a \emph{canonical set} of starting points that are representative of all possible starting points, and the size of this set is linear in the number of service requests.

Then, for each set of periods we use, we trim the windows and run a basic repairman algorithm.  For each set of trimmed windows, our analysis depends on an \emph{ensemble} of runs, each of which moves backward and forward along the path of an optimal run.  The ensemble is composed of a number of runs traveling at a speedup $s$ faster than an optimal run travels.  While it is too expensive to determine an actual optimal run unless \peqnp, our analysis will rely only on a definition of each run in the ensemble in terms of the optimal run.  One run in the ensemble might be effective in servicing requests serviced by an optimal run very early in their respective windows.  Another run might be effective in servicing requests serviced by an optimal run very late in their respective windows.  Together, these runs would service every request on trimmed windows that an optimal run would service on untrimmed windows.  An optimal run on trimmed windows, then, would achieve at least the average service of the ensemble of runs taken as a whole.  Likewise, the best run from among all possible choices of trimming would achieve at least the average service taken over all runs in the ensemble.

Except for the factor of $\gamma$ that does not come from trimming, our algorithm finds a run which performs as well as this average.  
Our analysis of an ensemble guarantees that a certain fraction of profit is always collected by the run produced by our algorithm.  Again, we know the ensemble of runs only symbolically, because we cannot efficiently know the instantiation of these runs unless \peqnp.  The challenge has been to devise ensembles so that our analysis will establish a good performance for our algorithm.

Although we are the first to extend time window problems into the realm of speedup, much other related work is being done with time windows.
For general metric spaces and general time windows together in the rooted problem, 
an $O(\log^2 n)$-approximation is given in \cite{Bansal}.  An $O(\log L)$-approximation is given in \cite{Chekuri3}, for the case that all time window start and end times are integers, where $L$ is the length of the longest time window.  In constrast, a constant approximation is given in \cite{Chekuri},
but only when there are a constant number of different time windows.  Following the initial publication of our work in \cite{Frederickson3}, an extension was given in \cite{Chekuri3} that gives an $O(\log D)$-approximation to the unrooted problem with general time windows, where $D$ is the ratio of the length of largest time window to the length of the smallest.  Polylogarithmic approximation algorithms to the directed traveling salesman problem with time windows have been given in \cite{Chekuri2} and \cite{Nagarajan2}.
The repairman problem with time windows has also been studied in the operations research community,
as in \cite{Focacci1} and \cite{Focacci2},
where it is exhaustively solved to optimality.

Note that the bicriteria approximation that we give will guarantee the repairman a greater fraction of the optimal profit by speeding up beyond the speed of the optimal service run.  Although we are the first to identify techniques that allow for a trade-off between requests serviced and speedup, there are other bicriteria approximations for repairman problems.  For example, a bicriteria approximation is presented in \cite{Bansal} that allows the repairman to service more requests if the time windows are all lengthened by some amount $\epsilon$.  If the time window for request $s_i$ starts at time $t$ and ends at time $t'$, this lengthening would change each time window from $[t, t')$ to $[t, (1 + \epsilon)t')$.  In contrast, the effect of running at speedup $s$ is equivalent to increasing all times by some factor, changing time window $[t, t')$ to $[st, st')$.

Before we explain our bicriteria algorithms in depth, we will summarize in Sect.~\ref{section:algorithms for repairman} our results from \cite{Frederickson3} as well as recent improvements.  In Sect.~\ref{section:ensemble}, we will explain the basic mechanisms needed for the averaging analysis used in our ensemble approach.  In Sect.~\ref{section:canonical}, we will introduce the concept of canonical sets needed to make our algorithms practical for arbitrary values of speedup $s$.  In Sect.~\ref{section:2<s<4}, we will give the analysis that proves our approximation results for speedup $s$ in the range $2 \leq s \leq 4$.  Finally, in Sect.~\ref{section:1<s<2}, we will give the more involved analysis proving approximation results for speedup $s$ in the range $1 \leq s \leq 2$.

\section{Algorithms for Repairman without Speedup}
\label{section:algorithms for repairman}

We continue to restrict our focus to the traveling repairman problem with unit-time windows.  For this problem, our repairman approximation algorithms in \cite{Frederickson3, Frederickson6} trim the time windows so that the resulting windows can be partitioned into non-overlapping periods.  These algorithms identify a variety of good paths within each separate period and then use dynamic programming to select and paste these paths together into a variety of longer good paths and ultimately a good service run for the whole problem.  Optimal paths within each period can be found in polynomial time when the underlying graph is a tree, but only approximately optimal paths are found on a metric graph.

To unify the notation between these approaches, we describe the approximation factor for finding paths within each period using $\gamma$, where $\gamma = 1$ for a tree and $\gamma = 2 + \epsilon$ for a metric graph.  Results from \cite{Frederickson3} clearly show that $\gamma \leq 5$ for a metric graph, but this result has been improved to $2 + \epsilon$ because of results in \cite{Chekuri2}.  To describe the running time for these repairman algorithms we bound the time by $\Gamma(n)$, where $\Gamma(n)$ is $O(n^4)$ for a tree and $O(n^{O(1/\epsilon^2)})$ for a metric graph, where the tree bound comes from results in \cite{Frederickson3} and the metric graph bound can be derived by augmenting those results with techniques from \cite{Chekuri2}.  A full explanation of this improved result is available in \cite{Frederickson6}.

\section{The Ensemble Approach for Analyzing Performance}
\label{section:ensemble}

The two results cited above give analysis for runs after trimming for the cases of no speedup ($s = 1$ guaranteeing $1/(3\gamma)$ of optimal profit) and full speedup ($s = 4$ guaranteeing $1/\gamma$ of optimal profit achieved at the original speed).  Our analysis cannot examine speeds lower than the reference speed of 1 (i.e, slowdown) because there is no guarantee that a slower run can ever earn any constant fraction of the profit of a faster run.  Likewise, above a speedup of 4, the effects of trimming are completely offset.    Thus, this range of speedups is the only range for which our bicriteria analysis seems practical.
For this range $1 < s < 4$, our analysis uses a number of different service runs that are defined by moving backward and forward along an optimal tour $R^*$.  We rely on averaging over a suitable ensemble of runs to establish that the run $R$ generated by our algorithm does well.

Let $R^*$ be an optimal service run at unit speed originally starting at time 0.
In  defining the service runs that we use in our analysis,
we use the term \emph{racing} to describe movement, forward and backward, along $R^*$ at a speedup of $s$ times the speed at which $R^*$ travels.  Please refer to Figure~\ref{figure:s=2} for an example of runs racing at $s = 2$.  For convenience, let $\delta = 1/(2s)$, which is the time needed to travel at speed $s$ the distance traveled in one period at unit speed.

One of the runs over which we average is service run $A$, defined as follows.  Start run $A$ at time $0$ at the location that $R^*$ has at time $-.5$.  Then run $A$ follows a pattern of racing forward along $R^*$ for $1$ period, racing backward along $R^*$ for $1 - 2\delta$ periods, and then racing forward along $R^*$ for $2\delta$ periods.  The positions reached by run $A$ at the times $0$, $.5$, $1$, $1.5$, and $2$ are indicated in Figure~\ref{figure:s=2}.  Note that the pattern of movement for run $A$ repeats every 2 periods.

Define $A^R$, the ``\emph{reverse}'' of $A$, as follows.  Start run $A^R$ at time $0$ at the location that $R^*$ has at time $.5$.  Then run $A^R$ follows a repeating pattern of racing forward along $R^*$ for $2\delta$ periods, racing backward along $R^*$ for $1 - 2\delta$ periods, and then racing forward along $R^*$ for $1$ period.  Similarly, the positions reached by run $A^R$ at times $0$, $.5$, $1$, $1.5$, and $2$ are indicated in the same figure.

If a service request $p$ is serviced by a run $R$ during the period that the time window of $p$ has been trimmed into, we say that $R$ \emph{covers} $p$.  We partition the set of requests whose windows are trimmed into a given period into three sets: Set $T$ contains those requests serviced by $R^*$ in that same period, set $E$ contains those requests serviced by $R^*$ in the preceding period, and set $L$ contains those requests serviced by $R^*$ in the following period.  See Figure~\ref{figure:s=2} for an example: Run $A$ services requests in $L$ for each period, and it services requests in $T$ for every second period.  Note, for example, that requests serviced on $R^*$ between times $0$ and $.5$ are serviced by $A$ during that same period.

\begin{figure}[!hbt]
\centering
\begin{xy}
\xyimport(397, 155){\includegraphics[width=.75\textwidth]{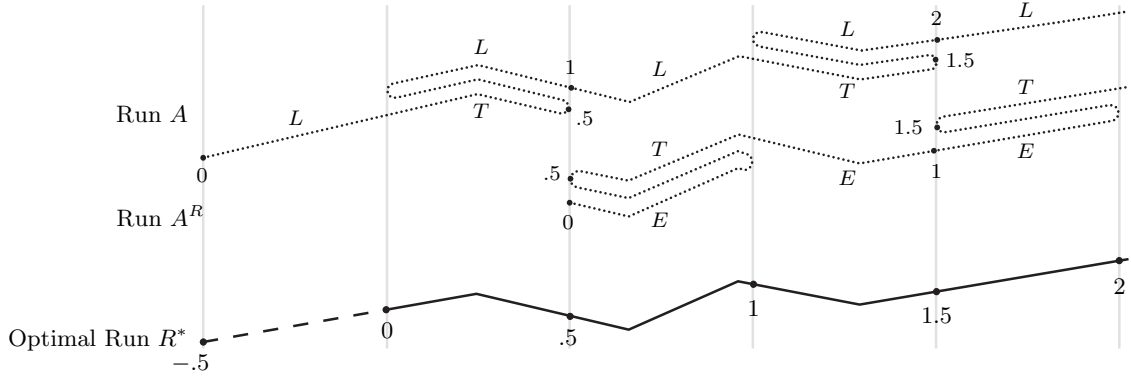}}
(2,60)*!R\txt\footnotesize{Run $A^R$};
(-5,4)*!R\txt\footnotesize{Optimal Run $R^*$};
(-5,106)*!R\txt\footnotesize{Run $A$};
%% Run A Periods
(41,104)*\txt\scriptsize{$L$};
(120,107)*\txt\scriptsize{$T$};
(120,136)*\txt\scriptsize{$L$};
(196,126)*\txt\scriptsize{$L$};
(276,116)*\txt\scriptsize{$T$};
(276,144)*\txt\scriptsize{$L$};
(352,153)*\txt\scriptsize{$L$};
%% Run B Periods
%(41,42)*\txt\scriptsize{$E$};
%(41,69)*\txt\scriptsize{$T$};
%(120,73)*\txt\scriptsize{$E$};
(196,58)*\txt\scriptsize{$E$};
(196,90)*\txt\scriptsize{$T$};
(276,77)*\txt\scriptsize{$E$};
(352,89)*\txt\scriptsize{$E$};
(352,118)*\txt\scriptsize{$T$};
%% Times
(80,8)*\txt\footnotesize{$0$};
(157,5)*\txt\footnotesize{$.5$};
(236,18)*\txt\footnotesize{$1$};
(314,15)*\txt\footnotesize{$1.5$};
%(-30,-6)*!L\txt\footnotesize{time $-.5$};
(-4,-7)*\txt\footnotesize{$-.5$};
(392,28)*\txt\footnotesize{$2$};
%% Run A \hat{t} times
%(-17,97)*!L\txt\footnotesize{time $0$};
(1,78)*\txt\scriptsize{$0$};
(164,104)*\txt\scriptsize{$.5$};
(158,127)*\txt\scriptsize{$1$};
(324,130)*\txt\scriptsize{$1.5$};
(314,149)*\txt\scriptsize{$2$};
%% Run B \hat{t} times
%(-3,55)*\txt\footnotesize{$-.5$};
%(-25,24)*!L\txt\footnotesize{time $-1$};
(157,57)*\txt\scriptsize{$0$};
(150,79)*\txt\scriptsize{$.5$};
(314,80)*\txt\scriptsize{$1$};
(302,100)*\txt\scriptsize{$1.5$};
\end{xy}
\caption{Runs $A$ and $A^R$ with a speedup of 2 based on an optimal run.  The pattern in $A^R$ perfectly complements $A$ to give full coverage.  Times are labeled on the optimal run as well as runs $A$ and $A^R$.  Segments of each run are also designated $L$, $T$, and $E$ depending on which subset of runs they add coverage to.  The light vertical lines indicate shared positions on the runs, not identical times.}
\label{figure:s=2}
\end{figure}

We number the periods $-1$, $0$, $1$, $2$, and so on by the integer multiples of $.5$ that give their starting times.  The labels $L$, $T$, and $E$ in Figure \ref{figure:s=2} make it clear that every period numbered 1 and later has full coverage.  The combination of runs $A$ and $A^R$ ensures that the location of a given request is visited three times: during the same period, the previous period, and the following period that the request was visited by $R^*$.  Thus, we guarantee that our combination of runs visits the request no matter which period its window was trimmed into.  However, it appears that periods $-1$ and $0$ only have partial coverage.  This apparent discrepancy is an endpoint anomaly.  In fact, period $-1$ needs no coverage because no locations visited by $R^*$ are visited in this period.  Likewise, there are no requests in set $E$ serviced by $R^*$ during period 0 because such requests would have been trimmed into period $-1$, which is impossible.

Although there might be some superficial resemblance between the $L$, $T$, and $E$ sets and the small margin and large margin cases discussed in \cite{Bansal}, the approach in that paper iteratively runs specialized algorithms that progressively redefine the meaning of small margins and large margins.  In contrast, the three sets that we use are fixed for any given trimming scheme, do not have separate algorithms tailored for each one, and are used purely for accounting.

Let $S$ be a subset of service requests.  Define the \emph{coverage} of $S$ by a run $R$, written $cover_R(S)$, to be the number of requests in $S$ satisfied by $R$ divided by the total number of requests in $S$.  Define the coverage of $S$ by a set $U$ of runs, written $cover_U(S)$, to be the average of $cover_R(S)$ for every run $R \in U$.

\begin{proposition}[Average Coverage] \label{proposition:averaging}  Let $\{S_1, S_2, S_3, ... S_a\}$ be a collection of sets of service requests such that $\bigcup_i S_i$ gives all the requests serviced by $R^*$ on untrimmed windows.  Let $U$ be a set of service runs.  If $\min_i\{cover_U(S_i)\} = \mu$, then there is at least one service run $\hat{R} \in U$ such that profit$(\hat{R}) \geq \mu \cdot$profit$( R^* )$.
\end{proposition}

The Average Coverage Proposition formalizes the following intuition.  Let a group of service runs achieve some coverage over a set of requests.  Let us also say that we have divided those requests into many different subsets, some of which may overlap, but the union of all the subsets is the original set of requests.  If we take the subset of requests with the worst average coverage, some service run in the group covers a fraction of total requests no smaller than that worst average coverage.  Otherwise, the average coverage of all subsets would be worse than the coverage of the worst covered subset, which is a contradiction.

Given a way of dividing requests into subsets, we wish to prove that some set of service runs achieves a particular lower bound on average coverage.  For the case of unit-length time windows, we divide requests into subsets based on membership in $L$, $T$, or $E$ sets (or smaller partitions of these three sets).  Since our dynamic programming approach finds an optimal run on trimmed windows, to within a factor of $\gamma$, we know that this run must perform at least as well as the bound we prove.

For analyzing $s = 2$ and $s = 3$, we use ensembles of runs.  We average the coverage over sets $L$, $T$, and $E$ by summing the lowest coverage for the requests in each set and dividing by the total number of runs used.  By the Average Coverage Proposition, this average gives a lower bound on the performance of an optimal run on trimmed intervals.

\begin{theorem}
\label{theorem:s=2}
Running a repairman algorithm from \cite{Frederickson3} at a speedup of $s = 2$ gives a $2\gamma$-approximation to repairman with unit-time windows in $\Gamma(n)$ time.
\end{theorem}

\begin{proof}
To analyze the performance of the resulting service run, we consider runs $A$ and $A^R$, as defined above.  Refer to Figure \ref{figure:s=2} for an example.
Run $A$ covers all of the requests in set $L$ and all of the requests in set $T$ which fall in even numbered periods with respect to $R^*$.  Run $A^R$ covers all of the requests in set $E$ and all of the requests in set $T$ which fall in odd numbered periods.  For each run, we record a 1 in each set for every period which is completed covered.  Note that $L_{\text{\emph{even}}}$ is that subset of requests in set $L$ that were serviced by $R^*$ in even numbered periods and $L_{\text{\emph{odd}}}$ is that subset of requests in set $L$ that were serviced by $R^*$ in odd numbered periods, and similarly for $T$ and $E$.  We sum these contributions into yields and then divide the yields by the total number of runs to get the average coverage for designation.  Table \ref{table:s=2} illustrates this process, in which we get an average coverage of 1/2.  Let $\hat{R}$ be the service run produced by the appropriate repairman algorithm from \cite{Frederickson3}.  By the Average Coverage Proposition, \emph{profit}$(\hat{R}) \geq \max\{$ \emph{profit}$(A)$, \emph{profit}$(A^R)\} \geq (1/2)$\emph{profit}$(R^*)/\gamma$, giving a $2\gamma$-approximation.\end{proof}

\begin{table}[!hbt]
\begin{center}
\begin{tabular}{ c c c c c c c c }
\toprule
 & & $L_{\text{\emph{even}}}$ & $L_{\text{\emph{odd}}}$ & $T_{\text{\emph{even}}}$ & $T_{\text{\emph{odd}}}$ & $E_{\text{\emph{even}}}$ & $E_{\text{\emph{odd}}}$\\
\toprule
Run $A$ & & 1 & 1 & 1 & 0 & 0 & 0 \\ 
\midrule
Run $A^R$ & & 0 & 0 & 0 & 1 & 1 & 1 \\ 
\toprule
Yield   & & 1 & 1 & 1 & 1 & 1 & 1\\ 
\toprule
Average Coverage & & 1/2 & 1/2 & 1/2 & 1/2 & 1/2 & 1/2 \\ 
\bottomrule
\end{tabular}
\end{center}
\caption{Averaging of runs $A$ and $A^R$ with speedup $s = 2$.}
\label{table:s=2}
\end{table}

\begin{theorem}
\label{theorem:s=3}
Running a repairman algorithm from \cite{Frederickson3} at a speedup of $s = 3$ finds a $4\gamma/3$-approximation to repairman with unit-time windows in $\Gamma(n)$ time.
\end{theorem}

\begin{proof}
To analyze the performance, we use four runs: $A$, $A^R$, $\vec{A}$, and $\vec{A}^R$, where $\vec{A}$ and $\vec{A}^R$ are shifted versions of $A$ and $A^R$.  Run $\vec{A}$ follows the same pattern as $A$ but starts the pattern at time $1/2$ at the location $R^*$ has at time 0.  Run $\vec{A}^R$ follows the same pattern as $A^R$ but starts the pattern at time $1/2$ at the location $R^*$ has at time $1$. Run $A$, now running with speedup $s = 3$, covers all of the requests in $L$ and $T$ but only half the requests in $E$.  Run $A^R$ covers all of the requests in $T$ and $E$ but only half the requests in $L$.  The shifted runs are copies of $A$ and $A^R$, covering all of $T$ and the other half of the requests of $E$ and $L$, respectively.
As before, we sum the contributions and divide by the total number of runs to find the average coverages shown in Table \ref{table:s=3}.  The worst coverage for any set, averaged over runs $A$, $A^R$, $\vec{A}$, and $\vec{A}^R$, is 3/4, corresponding to sets $L$ and $E$.  Let $\hat{R}$ be the service run produced by the appropriate repairman algorithm from \cite{Frederickson3}.  By the Average Coverage Proposition, \emph{profit}$(\hat{R}) \geq \max\{$ \emph{profit}$(A)$, \emph{profit}$(A^R)$, \emph{profit}$(\vec{A})$, \emph{profit}$(\vec{A}^R) \} \geq (3/4)$\emph{profit}$(R^*)/\gamma$, giving a $4\gamma/3$-approximation.\end{proof}

\begin{table}[!hbt]
\begin{center}
\begin{tabular}{ c c c c c c c c }
\toprule
 & & $L_{\text{\emph{even}}}$ & $L_{\text{\emph{odd}}}$ & $T_{\text{\emph{even}}}$ & $T_{\text{\emph{odd}}}$ & $E_{\text{\emph{even}}}$ & $E_{\text{\emph{odd}}}$\\
\toprule
 Run $A$ & & 1 & 1 & 1 & 1 & 1 & 0\\
\midrule
Run $\vec{A}$ & & 1 & 1 & 1 & 1 & 0 & 1\\  
\midrule
Run $A^R$ & & 0 & 1 & 1 & 1 & 1 & 1\\ 
\midrule
Run $\vec{A}^R$ & & 1 & 0 & 1 & 1 & 1 & 1\\ 
\toprule
Yield & & 3 & 3 & 4 & 4 & 3 & 3\\ 
\toprule
Average Coverage & & 3/4 & 3/4 & 1 & 1 & 3/4 & 3/4\\ 
\bottomrule
\end{tabular}
\end{center}
\caption{Averaging of runs $A$, $\vec{A}$, $A^R$, and $\vec{A}^R$ with speedup $s = 3$.}
\label{table:s=3}
\end{table}

\section{Canonical Collections and Their Use in Algorithms}
\label{section:canonical}

If speedup is a rational number, let $s = q/r$, where $q$ and $r$ are relatively prime.  Otherwise, if the speedup is an irrational number, we assume $r$ to be infinite.  Let $m$ be the number of requests.  Below we will define algorithm SPEEDUP such that it runs a repairman algorithm on each set of periods from among a collection of $\min\{r, m\}$ different sets of periods.  If $s$ is rational and $r \leq m$, we create sets of periods starting at $r$ time instants spread uniformly from time $0$ up to time $1/2$.   
If any window would start on any period or half period boundary,
we shift the periods by a small amount to avoid this case.
When $r>m$, SPEEDUP will use a canonical collection of sets of periods which is equally as effective as the $r$ sets (whenever $r$ is finite) but has size just linear in $m$.  The procedure CANONICAL defined below ``normalizes'' the beginnings of all the time windows into the range $(0, 1/2)$ and then picks a time between each consecutive pair of beginnings to be the starting point for a new set of periods.

\begin{table}[!hbt]
\begin{tabular}{l}
\smallskip \\
\toprule
\textbf{CANONICAL} \\
\midrule
\tab For time window $s_i$,\\
\tab \tab Let $t(s_i)$ be the starting time of $s_i$.\\
\tab \tab Let $a$ be the largest integer such that $a/2 < t(s_i)$.\\
\tab \tab Set $h_i$ to be $t(s_i) - a/2$.\\
\tab Sort the $h_i$ values, relabeling them $h_1, h_2, \ldots, h_m$ in increasing order.\\
\tab Let the first set of periods start at time 0.\\
\tab For $i$ from 2 to $m$,\\
\tab \tab Let the $i^{th}$ set of periods start at time $(h_{i - 1} + h_{i})/2$.\\
\bottomrule
\end{tabular}
\end{table}

\begin{lemma}
\label{lemma:m periods}
Any given set of periods is equivalent to one of the at most $m$ sets of periods
in the canonical collection, and this collection can be identified in $O(m \log m)$ time.
\end{lemma}

\begin{proof}
Without loss of generality, we assume that all time windows have unique starting times after normalization.  Label the normalized starting times of the windows $h_1$, $h_2$, $h_3$, $\ldots$, $h_m$ in ascending order as produced by the procedure CANONICAL.  The first point chosen by CANONICAL is 0.  The second point chosen is $(h_1 + h_2)/2$, the third is $(h_2 + h_3)/2$, and so on.

A starting point uniquely determines a set of periods, with each successive period starting at .5 time units after the start of the previous one.  Each set of periods induces a trimming of time windows according to the trimming procedure given before.  Define an equivalence relation on sets of periods:  Consider two trimmings induced by two different sets of periods.  If the difference in starting times from one set of periods to the next is the same for all time windows, then those two sets are equivalent.  Consider the time window starting at $h_i$.  Those sets of periods are equivalent that have starting times strictly between $h_i$ and $h_{i+1}$. Only by moving the starting time earlier than $h_i$ will the trimming of the earlier window (starting  at $h_i$) change, and only by moving the starting time later than $h_{i+1}$ will the trimming of the later window (starting at $h_{i+1}$) change.  Because our formulation of the traveling repairman problem does not specify a starting position or time, all equivalent sets of periods are functionally identical.  This equivalence relation partitions all possible sets of periods into equivalence classes.  By choosing a representative from each of these $m$ equivalence classes, we find a collection of $m$ sets of periods that covers all possible trimmings.

The running time for CANONICAL is $O(m \log m)$ because the $m$ starting times are sorted and all other work is linear.
\end{proof}

With CANONICAL defined, the description of SPEEDUP using speedup $s$ is straightforward.  We assume that no window starts on a period boundary as this situation can be easily remedied.  Let REPAIR be a basic repairman algorithm run after the trimming process has taken place, such as the algorithms in \cite{Frederickson3}.  As shown in \cite{Frederickson3}, REPAIR on a trimmed instance has an approximation ratio of at most $\gamma$ and a running time of $\Gamma(n)$.

\begin{table}[!hbt]
\begin{tabular}{l}
\smallskip \\
\toprule
\textbf{SPEEDUP} \\
\midrule
\tab If $s$ is rational and $s = q/r$ in reduced form, where $r < m$,\\
\tab \tab For $i$ from 1 to $r$,\\
\tab \tab \tab Let $h_i$ be $(i - 1)/r$.\\
\tab Otherwise,\\
\tab \tab Let $h_1$, $h_2$, $h_3$, \ldots $h_m$ be the period starting times found by CANONICAL.\\
\tab (Note that the total number of $h_i$ values is $\min\{r, m\}$.)\\
\tab For $i$ from 1 to $\min\{r, m\}$,\\
\tab \tab Start a set of periods at time $h_i$ and trim each window into the period it\\
\tab \tab completely fills.\\
\tab \tab Run REPAIR on the trimmed windows and retain the best result so far.\\
\bottomrule
\end{tabular}
\end{table}

Since algorithm SPEEDUP runs REPAIR a total of $\min\{r, m\}$ times, the running time for SPEEDUP is $O(\min\{r, m\}\Gamma(n))$.

\begin{observation}\label{observation:irrational}
Algorithm SPEEDUP works for any real number $1 \leq s \leq 4$; however, our analysis will assume that $s$ is a rational number such that $s = q/r$.  In the case that $s$ is irrational, our analysis still holds because the functions of $s$ with which we bound performance are piecewise smooth and continuous.  Thus, we can choose rational numbers arbitrarily close to any real number, and our analysis will be well-behaved in the limit.
\end{observation}

\section{Analysis of Speedup in the Range $2 \leq s \leq 4$}
\label{section:2<s<4}

In this section, we describe an ensemble of runs and the coverage given by such a set of runs when speedups $s$ is in the range $2 \leq s \leq 4$.  Based on the average coverage of the ensemble, the algorithm SPEEDUP produces a $4\gamma/s$-approximation in this speedup range.

We assume that no time window starts at time $i/(2r)$ for any integer $i$.  If any do, we can perturb all times slightly so that this is not the case.  Let $[w, w+1)$ be any time window.  Let $\omega$ be the smallest integer multiple of $1/(2r)$ that is greater than $w$.  We designate subintervals $[w, \omega)$, $[\omega, \omega + 1/(2r))$, $[\omega + 1/(2r), \omega + 2/(2r))$, $\ldots$, $[\omega + (2r - 1)/(2r), w + 1)$ by $w_0$, $w_1$, $w_2$, $\ldots$, $w_{2r}$.  When a given time window is trimmed, some portion of the original window will lie in an earlier period, half of it in the period it was trimmed into, and the remaining portion in the following period.  As before, the goal is to use an ensemble of runs to cover every part of each time window so that we can show a lower bound on profit with an averaging argument.  However, using sets $L$, $T$, and $E$ alone is too coarse-grained for general $s$.  Instead, we record the coverage of each subinterval $w_i$.  For accounting purposes, we assign a 1 for any subinterval which is covered every period, a $1/2$ for any subinterval covered every other period, and a 0 otherwise.  This accounting scheme implicitly averages runs $A$ and $A^R$ with their respective shifted versions $\vec{A}$ and $\vec{A}^R$.  Because the union of requests serviced during these subintervals consists of all the requests serviced by an optimal path, the Average Coverage Proposition still applies.

We partition sets $L$, $T$, and $E$ into subsets: $L_1$ through $L_{r}$, $T_1$ through $T_{r}$, and $E_1$ through $E_{r}$ respectively.  If we partition each period into $r$ equal-length divisions, then subset $L_j$, $T_j$, or $E_j$, for any given $j$, will consist of those service requests from set $L$, $T$, or $E$, respectively, which were serviced by the optimal run in the $j^{th}$ division of a period.  Note that, for a given period starting time, the window for a service request can only overlap with $2r + 1$ of these divisions.  Thus, we can make a one-to-one mapping of the $2r + 1$ subintervals of a window to $2r + 1$ of the $3r$ subsets of $L$, $T$, or $E$.  For a given window, some set of periods maps all subintervals of the window to $L_{1}$ through $L_{r}$, $T_1$ through $T_{r}$, and $E_1$.  Another set of periods maps all subintervals of the window to $L_{2}$ through $L_{r}$, $T_1$ through $T_{r}$, and $E_1$ through $E_2$.  Each remaining set of periods maps one fewer $L_i$ subset and one additional $E_j$ subset, where $j = i + 1$ for set $i$.  Table \ref{table:2<s<4} shows which subsets are mapped to subintervals of some time window depending on the set of periods used.

\begin{table}[!hbt]
\begin{center}
\begin{tabular}{ c c c c c c c c c c c c c c }
\toprule
Set of & & & \multicolumn{11}{c}{Subintervals} 
\\
%\hline
  Periods & & & $w_0$     & $w_1$ & $w_2$ & $\ldots$ & $w_{r}$ & $w_{r + 1}$ & $w_{r + 2}$ & $\ldots$ & $w_{2r - 2}$ &  $w_{2r - 1}$ & $w_{2r}$ \\
\toprule
	    $0$ & & & $L_1$ & $L_2$ & $L_3$ & $\ldots$ & $T_1$ & $T_2$ & $T_3$                      & $\ldots$ & $T_{r - 1}$     &  $T_{r}$     & $E_1$ \\
\midrule
	    $1$ & & & $L_2$ & $L_3$ & $L_4$ & $\ldots$ & $T_2$ & $T_3$ & $T_4$                      & $\ldots$ & $T_{r}$     &  $E_1$     & $E_2$ \\
\midrule
	    $\vdots$ & & & $\vdots$ & $\vdots$ & $\vdots$ &  & $\vdots$ & $\vdots$ & $\vdots$         &  & $\vdots$     &  $\vdots$     & $\vdots$ \\
\midrule
	    $r - 2$ & & & $L_{r - 1}$ & $L_{r}$ & $T_1$ & $\ldots$ & $T_{r - 1}$ & $T_{r}$ & $E_1$         & $\ldots$ & $E_{r - 3}$     &  $E_{r - 2}$     & $E_{r - 1}$ \\
\midrule
	    $r - 1$ & & & $L_{r}$ & $T_1$ & $T_2$ & $\ldots$ & $T_{r}$ & $E_1$         & $E_2$         & $\ldots$ & $E_{r - 2}$     &  $E_{r - 1}$     & $E_{r}$ \\
\bottomrule
\end{tabular}\\
\end{center}
\caption{Possible ways to map the subintervals of a time window to subsets.}
\label{table:2<s<4}
\end{table}

Figure~\ref{figure:speeds} gives three examples of type $A$ runs
for various speedups in the range $2<s<4$.  Note that all of the runs in this figure arrive at the same position at the beginning of every second period, namely at times $0, 1, 2,$ and so on.  Portions of runs servicing requests in $L$, $T$, and $E$ or various subsets are identified: subset $E_1$ being those service requests of $E$ serviced by $R^*$ in the first half of a period for $s = 7/2$, subsets $T_1$, $T_2$, $T_3$, $E_1$, $E_2$, and $E_3$ being those service requests of $T$ or $E$ serviced by $R^*$ in various quarter periods for $s = 11/4$, and subsets $T_1$, $T_2$, $E_1$, and $E_2$ being those service requests of $T$ or $E$ serviced by $R^*$ in various thirds of periods for $s = 7/3$.

\begin{figure}[!hbt]
\centering
\begin{xy}
\newxycolor{white}{1 gray}
\xyimport(397, 227){\includegraphics[width=.75\textwidth]{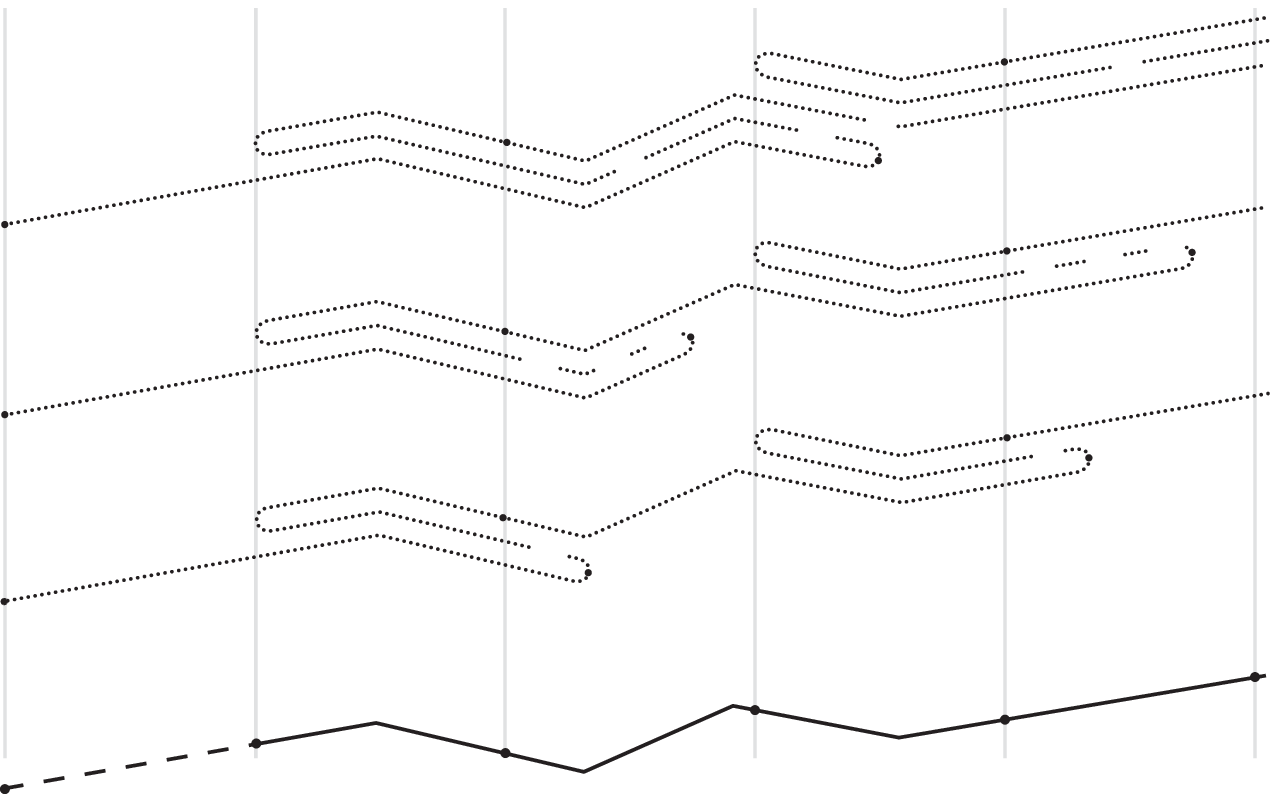}}
(-24,185)*!R\txt\footnotesize{Run $A$ for};
(-24,170)*!R{\large s = {7 \over 2}};
(-24,128)*!R\txt\footnotesize{Run $A$ for};
(-24,113)*!R{\large s = {11 \over 4}};
(-24,75)*!R\txt\footnotesize{Run $A$ for};
(-24,60)*!R{\large s = {7 \over 3}};
(-24,1)*!R\txt\footnotesize{Optimal Run $R^*$};
%% Optimal Times
(80,6)*\txt\footnotesize{$0$};
(157,3)*\txt\footnotesize{$.5$};
(236,15)*\txt\footnotesize{$1$};
(314,12)*\txt\footnotesize{$1.5$};
(-4,-8)*\txt\footnotesize{$-.5$};
(392,25)*\txt\footnotesize{$2$};
%% 7/3 Times
(1,47)*\txt\scriptsize{$0$};
(190,61)*\txt\scriptsize{$.5$};
(157,87)*\txt\scriptsize{$1$};
(350,98)*\txt\scriptsize{$1.5$};
(314,110)*\txt\scriptsize{$2$};
%% 11/4 times
(1,100)*\txt\scriptsize{$0$};
(222,134)*\txt\scriptsize{$.5$};
(157,141)*\txt\scriptsize{$1$};
(382,157)*\txt\scriptsize{$1.5$};
(314,164)*\txt\scriptsize{$2$};
%% 7/2 times
(1,155)*\txt\scriptsize{$0$};
(281,182)*\txt\scriptsize{$.5$};
(157,195)*\txt\scriptsize{$1$};
(314,218)*\txt\scriptsize{$2$};
%% 7/3 Labels
(41,69)*\txt\scriptsize{$L$};
(119,67)*\txt\scriptsize{$T$};
(172,55)*\txt\tiny{$E_1$};
(172,70)*\txt\tiny{$T_1$};
(119,93)*\txt\scriptsize{$L$};
(196,86)*\txt\scriptsize{$L$};
(276,77)*\txt\scriptsize{$T$};
(328,84)*\txt\tiny{$E_1$};
(328,98)*\txt\tiny{$T_1$};
(276,105)*\txt\scriptsize{$L$};
(352,114)*\txt\scriptsize{$L$};
%% 11/4 Labels
(41,122)*\txt\scriptsize{$L$};
(119,120)*\txt\scriptsize{$T$};
(170,110)*\txt\tiny{$E_1$};
(192,110)*\txt\tiny{$E_2$};
(208,116)*\txt\tiny{$E_3$};
(170,124)*\txt\tiny{$T_1$};
(192,124)*\txt\tiny{$T_2$};
(208,130)*\txt\tiny{$T_3$};
(119,147)*\txt\scriptsize{$L$};
(196,139)*\txt\scriptsize{$L$};
(276,131)*\txt\scriptsize{$T$};
(325,138)*\txt\tiny{$E_1$};
(346,142)*\txt\tiny{$E_2$};
(365,145)*\txt\tiny{$E_3$};
(325,151)*\txt\tiny{$T_1$};
(346,154)*\txt\tiny{$T_2$};
(365,157)*\txt\tiny{$T_3$};
(276,158)*\txt\scriptsize{$L$};
(352,167)*\txt\scriptsize{$L$};
%% 7/2 Labels
(41,177)*\txt\scriptsize{$L$};
(119,175)*\txt\scriptsize{$T$};
(196,167)*\txt\scriptsize{$E$};
(196,181)*\txt\scriptsize{$T$};
(119,202)*\txt\scriptsize{$L$};
(196,194)*\txt\scriptsize{$L$};
(256,190)*\txt\tiny{$E_1$};
(352,197)*\txt\scriptsize{$E$};
(276,194)*\txt\scriptsize{$T$};
(276,212)*\txt\scriptsize{$L$};
(352,209)*\txt\scriptsize{$T$};
(352,222)*\txt\scriptsize{$L$};
\end{xy}
\caption{Examples of type $A$ runs for three different speedups in the range $2 < s < 4$, namely at $7/3$, $11/4$, and $7/2$.}
\label{figure:speeds}
\end{figure}

Algorithm SPEEDUP runs a repairman algorithm on $\min \{r, m\}$ different sets of periods and takes the best run from those choices.

\begin{theorem}
For any $s$ in the range $2 \leq s \leq 4$, SPEEDUP finds a $4\gamma/s$-approximation to repairman with unit-time windows in $O(\min\{r, m\}\Gamma(n))$ time.
\label{thm:W3for2<s<4}
\end{theorem}

\begin{proof}
For any window $[w, w + 1)$, on service run $A$, if a trimmed interval starts with subinterval $w_i$, where $1 \leq i \leq r$, then the window's first $q - 2r + i$ subintervals will be covered each period, and the next $r$ subintervals after that will be covered every other period (but no further than the last subinterval).  When $s < 3$, run $A$ completely covers subsets $L_1$ through $L_{r}$, covers those requests in subsets $T_1$ through $T_{r}$ serviced by $R^*$ in half of the periods, and covers those requests in subsets $T_1$ through $T_{q - 2r}$ and subsets $E_1$ through $E_{q - 2r}$ serviced by $R^*$ in the other half of the periods.  When $s > 3$, run $A$ completely covers subsets $L_1$ through $L_{r}$ and subsets $T_1$ through $T_{r}$, covers those requests in subsets $E_1$ through $E_{q - 3r}$ serviced by $R^*$ in half of the periods, and covers those requests in subsets $E_1$ through $E_{r}$ serviced by $R^*$ in the other half of the periods.  
A similar analysis applies to $\vec{A}$, $A^R$, and $\vec{A}^R$.

We assign a 1 for any subinterval which is covered every period and a $1/2$ for any subinterval covered every other period.  For runs $A$ and $\vec{A}$ together, $w_i$ earns a 1 for all $r$ sets of periods where $0 \leq i \leq q - 2r$, giving a total of $r$ for each such $i$.  For each $i > q - 2r$, the total decreases by $1/2$ from the total for $i - 1$.   For runs $A^R$ and $\vec{A}^R$ together, $w_i$ gets 1 for all $r$ sets of periods where $2r - (q - 2r) \leq i \leq 2r$, giving a total of $r$ for each such $i$.  For each $i < 2r - (q - 2r)$, the total decreases by $1/2$ from the total for $i + 1$.

Now, we take the total over all sets of periods for both runs.  For runs $A$ and $\vec{A}$ together we get a total of $r$ for $w_0$.  For runs $A^R$ and $\vec{A}^R$ together we get a total of $r - (1/2)(2r - (q - 2r)) = q/2 - r$.  Summing these together, we get a yield of $q/2$ for $w_0$.  By symmetry, the total for $w_{2r}$ is also $q/2$.  In the case that $s < 3$, the contributions for runs $A$ and $\vec{A}$ together decrease by no more than $1/2$ as the contributions for runs $A^R$ and $\vec{A}^R$ together increase by $1/2$ for each increment of $i$ in the range $0 < i \leq r$.  In the case that $s \geq 3$,  contributions from $A$ and $\vec{A}$ together are constant and contributions from $A^R$ and $\vec{A}^R$ together only increase or stay constant for $0 < i \leq r$.  Thus, the yield for all other $w_i$ in all cases is at least $q/2$.  Since $r$ sets of periods for the two pairs of runs cost a total of $2r$ sets of periods to average over, the yield for the algorithm is at least $q/2 \cdot 1/(2r) = q/(4r) = s/4$.  Multiplying the reciprocal by $\gamma$ gives a $4\gamma/s$-approximation.  By Observation \ref{observation:irrational}, this result holds for all real values of $s$.  Note that the details of this proof can also be generated using the CREATE-TABLE procedure given in Sect.~\ref{section:1<s<2}.\end{proof}

\section{Analysis of Speedup in the Range $1 \leq s \leq 2$}
\label{section:1<s<2}

A different analysis gives us better performance as a function of $s$ in the range $1 \leq s \leq 2$.  Based on the average coverage of the ensemble, the algorithm SPEEDUP produces a $6\gamma/(s + 1)$-approximation in this speedup range.  We consider four runs (and their shifts) and balance their relative weightings.  For the range $1 \leq s \leq 2$, we represent speedup $s$ as $s = (r + k)/r$ with integers $r \geq 1$ and $0 \leq k \leq r$.

Our analysis will require several versions of $A$ that have different starting points.  To simplify notation, let a \emph{hop} be the distance traveled in time $1/(2r)$ at unit speed.  Let $A_\Delta$ be run $A$ starting $\Delta$ hops further along and run $A_\Delta^R$ be run $A^R$ starting $\Delta$ hops further back. Run $A_\Delta$ follows the same pattern of movement as run $A$ but starts at $t = 0$ at the location that $R^*$ has at  $t = -1/2 + \Delta/(2r)$.  Its reverse $A_\Delta^R$ follows the same pattern of movement as $A^R$ but starts at $t = 0$ at the location that $R^*$ has at $t = 1/2 - \Delta/(2r)$.

Examples of runs $A$, $A^R$, $A_{r - k}$, and $A_{r - k}^R$ are shown in Figure~\ref{figure:s=5/4} for $s = 5/4$ where $r = 4$ and $k = 1$.  Portions of runs servicing requests in $L$, $T$, and $E$ or various subsets are identified: subsets $L_4$, $T_1$, $T_2$, $T_3$, $T_4$, and $E_1$ being those service requests of $L$, $T$, and $E$ serviced by $R^*$ in various quarter periods.  We also define $\vec{A}$, $\vec{A}^R$, $\vec{A}_{r - k}$, and $\vec{A}_{r - k}^R$ as indicated earlier and implicitly average the contributions of these shifted runs with the contributions of their unshifted versions in Tables \ref{table:lowcontributions} and \ref{table:highcontributions}.

\begin{figure}[!hbt]
\centering
\begin{xy}
\xyimport(366, 236){\includegraphics[width=.75\textwidth]{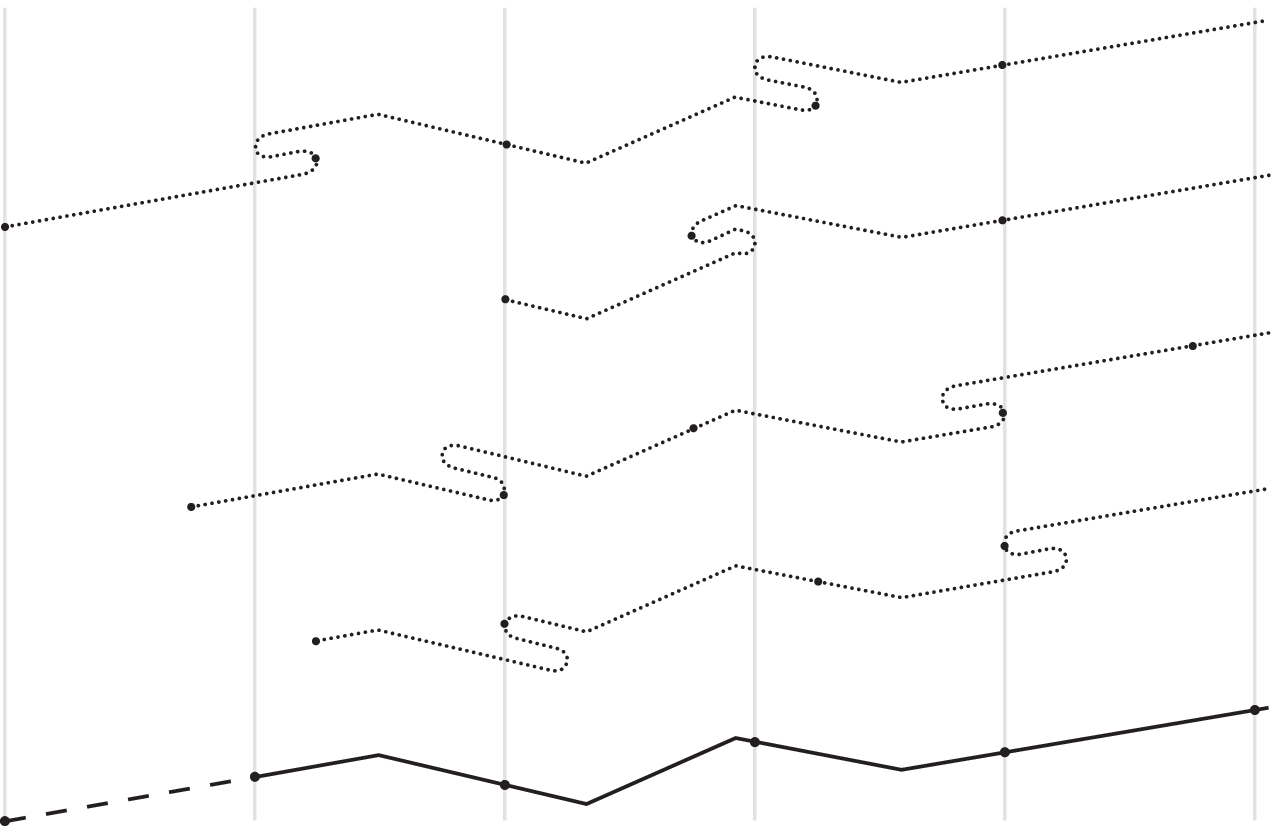}}
(-8,180)*!R\txt\footnotesize{Run $A$};
(-2,140)*!R\txt\footnotesize{Run $A^R$};
(-8,100)*!R\txt\footnotesize{Run $A_{r - k}$};
(-2,60)*!R\txt\footnotesize{Run $A_{r - k}^R$};
(-6,2)*!R\txt\footnotesize{Optimal Run $R^*$};
%% Optimal Times
(-2,-8)*\txt\footnotesize{$-.5$};
(73,6)*\txt\footnotesize{$0$};
(144,3)*\txt\footnotesize{$.5$};
(218,16)*\txt\footnotesize{$1$};
(290,13)*\txt\footnotesize{$1.5$};
(362,25)*\txt\footnotesize{$2$};
%% Run D Times
(84,52)*\txt\scriptsize{$0$};
(137,62)*\txt\scriptsize{$.5$};
(236,79)*\txt\scriptsize{$1$};
(278,79)*\txt\scriptsize{$1.5$};
%% Run C Times
(48,91)*\txt\scriptsize{$0$};
(152,93)*\txt\scriptsize{$.5$};
(200,107)*\txt\scriptsize{$1$};
(300,121)*\txt\scriptsize{$1.5$};
(344,131)*\txt\scriptsize{$2$};
%% Run B Times
%(-2,119)*\txt\footnotesize{$-1$};
%(42,146)*\txt\footnotesize{$-.5$};
(139,153)*\txt\scriptsize{$0$};
(191,168)*\txt\scriptsize{$.5$};
(289,167)*\txt\scriptsize{$1$};
%% Run A Times
(1,165)*\txt\scriptsize{$0$};
(98,194)*\txt\scriptsize{$.5$};
(146,189)*\txt\scriptsize{$1$};
(245,206)*\txt\scriptsize{$1.5$};
(289,212)*\txt\scriptsize{$2$};
%% Run D Labels
(100,48)*\txt\tiny{$T_2$};
(118,48)*\txt\tiny{$T_3$};
(136,44)*\txt\tiny{$T_4$};
(154,40)*\txt\tiny{$E_1$};
(181,68)*\txt\scriptsize{$T$};
(227,66)*\txt\tiny{$E_1$};
(245,63)*\txt\tiny{$T_2$};
(263,60)*\txt\tiny{$T_3$};
(281,63)*\txt\tiny{$T_4$};
(299,66)*\txt\tiny{$E_1$};
(326,97)*\txt\scriptsize{$T$};
%% Run C Labels
(64,100)*\txt\tiny{$L_4$};
(109,95)*\txt\scriptsize{$T$};
(137,115)*\txt\tiny{$L_4$};
(154,110)*\txt\tiny{$T_1$};
(172,110)*\txt\tiny{$T_2$};
(190,117)*\txt\tiny{$T_3$};
(208,125)*\txt\tiny{$L_4$};
(254,119)*\txt\scriptsize{$T$};
(281,134)*\txt\tiny{$L_4$};
(299,137)*\txt\tiny{$T_1$};
(317,140)*\txt\tiny{$T_2$};
(335,143)*\txt\tiny{$T_3$};
(353,145)*\txt\tiny{$L_4$};
%% Run B Labels
(181,144)*\txt\scriptsize{$E$};
(208,184)*\txt\tiny{$T_4$};
(254,178)*\txt\scriptsize{$E$};
(326,187)*\txt\scriptsize{$E$};
%% Run A Labels
(36,185)*\txt\scriptsize{$L$};
(82,180)*\txt\tiny{$T_1$};
(109,211)*\txt\scriptsize{$L$};
(181,203)*\txt\scriptsize{$L$};
(227,200)*\txt\tiny{$T_1$};
(254,221)*\txt\scriptsize{$L$};
(326,232)*\txt\scriptsize{$L$};
\end{xy}
\caption{Examples of runs $A$, $A^R$, $A_{r - k}$, and $A_{r - k}^R$ with a speedup of $s = (r + k)/r = 5/4$.}
\label{figure:s=5/4}
\end{figure}

In order to create these tables, we define a procedure called CREATE-TABLE that describes the process of determining coverage for a particular speedup $s$ for a particular run moved $\Delta$ hops.  Note that CREATE-TABLE is not an algorithm that is run in the process of finding an approximation to a repairman problem with speedup.  Rather, it provides a template that can be used to produce the tables used in analyzing the performance of such approximations.

Before CREATE-TABLE can be completely defined, it is necessary to explain the pattern of coverage generated by a run.  For all type $A$ runs, this pattern takes one of two forms.  Let $s$ be a rational number such that $s = q/r$.  Type $A$ runs repeat every two periods and thus can be represented with a pattern of coverage that uses a 1 to signify a subinterval covered every period and a 1/2 to signify a subinterval covered every other period.

Observe the movement of type $A$ runs, noting that such a run moves forward in one period the same distance that run $R^*$ moves in $q$ subintervals.  During the next period of time, a type $A$ run moves backward the same distance that run $R^*$ moves during $q - r$ subintervals and then forward the same distance that run $R^*$ moves during $r$ subintervals. Then, the pattern repeats.  When $s < 2$, run $A$ during its first period in the pattern covers $q$ successive subintervals as it moves forward, while in its second period covers $r$ subintervals as it moves forward.  Note that those subintervals covered as $A$ moves backward add nothing additional to the coverage.  Thus, this pattern of coverage is represented as $r$ repetitions of 1 and $q - r$ repetitions of 1/2.  %When $s \geq 2$, the pattern of coverage somewhat different, as discussed in the previous section.

\begin{table}[!hbt]
\begin{tabular}{l}
\smallskip \\
\toprule
\textbf{CREATE-TABLE(hops $\Delta$)} \\
\midrule
\tab Let the first element of the coverage pattern be indexed at 0.\\
\tab Let the subintervals of interest be numbered 0 through $2r$.\\
\tab Define function $\mathcal{C}$ based on the coverage pattern, such that:\\
\tab \tab $\mathcal{C}(i) = \left\{ \begin{array}{ll}
1&\mbox{if term }i\mbox{ of the pattern is 1}\\ 
1/2&\mbox{if term }i\mbox{ of the pattern is 1/2}\\ 
0&\mbox{otherwise}
\end{array} \right.$ \\
\tab Define function ${\mathcal F}$ on integers $i$, where $0 \leq i \leq 2r$:\\
\tab \tab ${\mathcal F}(i) = \sum_{j = 0}^{r - 1} \mathcal{C}(i + j - \Delta)$.\\
\tab Define function ${\mathcal F}^R$ on the same domain:\\
\tab \tab ${\mathcal F}^R(i) = {\mathcal F}(2r - i)$.\\
\tab The final coverage function defined by the table is given by ${\mathcal F}(i) + {\mathcal F}^R(i)$.\\
\bottomrule
\end{tabular}
\end{table}

CREATE-TABLE defines functions ${\mathcal F}(i)$ and ${\mathcal F}(i)^R$ such that the final coverage that we define is ${\mathcal F}(i) + {\mathcal F}(i)^R$ for $i = 0, 1, \ldots, 2r$.  Note that the functions $\mathcal{F}$ and $\mathcal{F}^R$ are piecewise linear functions with ranges dependent on the fundamental pattern of coverage.  Due to its construction, the combination ${\mathcal F}(i) + {\mathcal F}^R(i)$ is also a piecewise linear function and symmetrical.  Thus, only the range $0 \leq i \leq r$ need be listed in tables.\\

Although CREATE-TABLE gives a procedure for creating a table for a given speedup, we need tables expressed symbolically to prove coverage for a range of speedups.  Instead of using specific numbers, we can leave the basic patterns of subinterval coverage for a given style of run (such as type $A$ runs) with its shifted version in terms of $q$ and $r$.  By shifting this pattern $r$ times and summing the results together, we account for the different alignments a time window might have with respect to the various subintervals.  This sum is the function $\mathcal{F}$, which we can now express as a piecewise linear function parameterized by $q$ and $r$.  Function $\mathcal{F}^R$, which describes type $A^R$ runs, can be similarly described.  To combine the two functions symbolically, we sort the end points of the subinterval ranges from both descriptions together.  If, for the given range of speedups being considered, there are two end points which cannot be ordered, we subdivide the range of speeds so that, in each new speed range, the two end points in question can be ordered.  Once the end points of each subinterval range have been sorted, combining the descriptions from the normal and reversed functions of the runs is achieved by simply summing each range.

\begin{table}[!hbt]
\begin{center}
\begin{tabular}{rlr}
$\mathcal{F}(i) = $ & $\left\{ \begin{array}{l}
r -  {1 \over 2}i \smallskip \\ 
r + {1 \over 2}k - i \smallskip \\ 
{1 \over 2}r + {1 \over 2}k - {1 \over 2}i
\end{array} \right.$ &
$\begin{array}{rcl}
0 \leq &i& \leq k \smallskip \\ 
k \leq &i& \leq r \smallskip \\
r \leq &i& \leq r + k
\end{array}$
\\ \\
$\mathcal{F}^R(i) = $ & $\left\{ \begin{array}{l}
{1 \over 2}i - {1 \over 2}r + {1 \over 2}k \smallskip \\ 
{1 \over 2}k - r + i \smallskip \\ 
{1 \over 2}i
\end{array} \right.$ &
$\begin{array}{rcl}
r - k \leq &i& \leq r \smallskip \\ 
r \leq &i& \leq 2r - k \smallskip \\
2r - k \leq &i& \leq 2r
\end{array}$
\end{tabular}
\end{center}
\caption{Separate coverage functions for runs $A$ and $A^R$ when $1 \leq s \leq 2$.}
\label{table:F functions for A runs}
\end{table}

In Table \ref{table:F functions for A runs} we use $A$ runs to give an example of how symbolic coverage works.
We are only interested in the range $0 \leq i \leq r$.  In this range, the sub-ranges $0 \leq i \leq k$ and $k \leq i \leq r$ for $\mathcal{F}$ overlap with the sub-ranges $0 \leq i \leq r - k$ and $r - k \leq i \leq r$ for $\mathcal{F}^R$.  When $k \leq r - k$, for $k \leq i \leq r - k$, $\mathcal{F}(i) + \mathcal{F}^R(i) = (r + k/2 - i) + 0 = r + k/2 - i$, as in Table \ref{table:lowcontributions}. When $k \geq r - k$, for $r - k \leq i \leq k$, $\mathcal{F}(i) + \mathcal{F}^R(i) = (r - i/2) + (i/2 - r/2 + k/2) = r/2 + k/2$, as in Table  \ref{table:highcontributions}.  The combined coverages of runs $A$, $A^R$, $A_{r-k}$, $A_{r-k}^R$, and all of their respective shifted versions are all given in Table \ref{table:lowcontributions} when $k \leq r - k$ and in Table \ref{table:highcontributions} when $k \geq r - k$.

\begin{table}[!hbt]
\begin{center}
\begin{tabular}{rcll}
\begin{tabular}{r}
Combined contributions\\
for $A$ and $A^R$
\end{tabular} & = &
$\left\{ \begin{array}{l}
r - {1 \over 2}i \smallskip \\ 
r + {1 \over 2}k - i \smallskip \\ 
{1 \over 2}r + k  - {1 \over 2}i
\end{array} \right.$ &
$\begin{array}{rcl}
0 \leq &i& \leq k \smallskip \\ 
k \leq &i& \leq r - k \smallskip \\
r - k \leq &i& \leq r
\end{array}$
\\ \\
\begin{tabular}{r}
Combined contributions\\
for $A_{r - k}$ and $A_{r - k}^R$
\end{tabular} & = &
$\left\{ \begin{array}{l}
k + {3\over 2}i \smallskip \\
{1 \over 2}k + 2i \smallskip \\
{3 \over 2}r - k  + {1 \over 2}i
\end{array} \right.$ &
$\begin{array}{rcl}
0 \leq &i&  \leq k \smallskip \\
k \leq &i&  \leq r - k \smallskip \\
r - k \leq &i& \leq r
\end{array}$
\end{tabular}\\
\end{center}
\caption{Combined contributions of runs $A$, $A^R$, $A_{r - k}$, $A_{r - k}^R$, and all of their respective shifted versions, for $1\leq s \leq 2$ for the $i$th interval when $k \leq r - k$.  Does not include values for $i > r$ because those values are symmetrical.}
\label{table:lowcontributions}
\end{table}
\label{ptablc}

\begin{table}[!hbt]
\begin{center}
\begin{tabular}{rcll}
\begin{tabular}{r}
Combined contributions\\
for $A$ and $A^R$ 
\end{tabular} & = & 
$\left\{ \begin{array}{l}
r - {1 \over 2}i \smallskip \\
{1 \over 2}r + {1 \over 2}k \smallskip \\
{1 \over 2}r + k - {1 \over 2}i
\end{array} \right.$ &
$\begin{array}{rcl}
0 \leq &i& \leq r - k \smallskip \\
r - k \leq &i& \leq k \smallskip \\
k \leq &i& \leq r
\end{array}$
\\ \\
\begin{tabular}{r}
Combined contributions\\
for $A_{r - k}$ and $A_{r - k}^R$
\end{tabular} & = &
$\left\{ \begin{array}{l}
k + {3 \over 2}i \smallskip \\
{3 \over 2}r - {1 \over 2}k \smallskip \\
{3 \over 2}r - k + {1 \over 2}i\\
\end{array} \right.$ &
$\begin{array}{rcl}
0 \leq &i&  \leq r - k \\
r - k \leq &i& \leq k\\
k \leq &i& \leq r\\
\end{array}$
\end{tabular}\\
\end{center}
\caption{Combined contributions of the same runs as in Table~\ref{table:lowcontributions}, when $k \geq r - k$.}
\label{table:highcontributions}
\end{table}
\label{ptabhc}

\begin{lemma}\label{lemma:lowestinterval}
When we weight the combined contributions from $A$ and $A^R$ twice as much as the combined contributions from $A_{r - k}$ and $A_{r - k}^R$,
the yield for all intervals is at least $2r + k$.
\end{lemma}

\begin{proof}
We consider yields only for subintervals $w_i$ with $0 \leq i \leq r$ because the cases with $r \leq i \leq 2r$ are symmetrical.  We first handle the case when $k \leq r - k$, consulting Table \ref{table:lowcontributions}.

If $0 \leq i \leq k$, then the yield for $w_i$ is $2r + k + i/2$, which is at least $2r + k$, since $i \geq 0$.

If $k \leq i \leq r - k$, then the yield for $w_i$ is $2r + 3k/2$, which is greater than $2r + k$.

If $r - k \leq i \leq r$, then the yield for $w_i$ is $5r/2 + k - i/2$, which is at least $2r + k$, since $i \leq r$.

We now handle the case when $k \geq r - k$, consulting Table \ref{table:highcontributions}.
The algebra for the cases when $0 \leq i \leq r-k$ and $k \leq i \leq r$ gives exactly the same results as the first and third ranges from the previous part of the proof.  If $r - k \leq i \leq k$,  then the yield for subinterval $i$ is $5r/2+ k/2$, which is at least $2r + k$, since $r \geq k$.
\end{proof}

\begin{theorem}
For speedup $s$ in the range $1 \leq s \leq 2$, algorithm SPEEDUP finds a $6\gamma/(s+1)$-approximation to repairman with unit-time windows in $O(\min\{r, m\}\Gamma(n))$ time.
\label{thm:W3for1<s<2}
\end{theorem}

\begin{proof}
By Lemma \ref{lemma:lowestinterval}, the yield is at least $2r + k$.  Since we use two copies each of $A$ and $A^R$ and a single copy each of $A_{r - k}$ and $A_{r - k}^R$ averaged over $r$ different sets of periods, we apply the Average Coverage Proposition over $6r$ runs.  Thus, the approximation ratio is no more than
$6r\gamma /(2r + k) = 6r\gamma/((r + k) + r) = 6\gamma/(s + 1)$.  By Observation \ref{observation:irrational}, this result holds for all real values of $s$.
\end{proof}

\section{Conclusion}

As shown in \cite{Frederickson6}, the technique of trimming provides a powerful tool in attacking route-planning problems with time windows.  With trimming, it is possible to simplify the structure of many such problems and apply an ordering that allows dynamic programming to work.  For unrooted problems, the cost of this additional order is at most a constant reduction in the profit a run can earn.  In this paper we have shown that this reduction in profit can be offset, in part or in whole, by allowing some corresponding increase in speed over a hypothetical optimal benchmark.  Although we have only given the techniques needed for unit-time windows here, the extension to the more complicated case of window lengths over a range of values is forthcoming in \cite{Frederickson4}.\\

\noindent\textbf{Acknowledgments:} We thank Nelson Uhan for bringing a related reference on scheduling to our attention.

\singlespacing
\bibliography{bibliography}
\end{document}